\documentclass[12 pt]{extarticle}
\usepackage{amsmath}
\usepackage{amssymb}
\usepackage[a4paper, total={6in, 9in}]{geometry}
\usepackage{amsthm,url,tikz}
\usepackage[utf8]{inputenc}
\usepackage[english]{babel}
\usepackage{graphicx}
\usepackage{multirow} 
\usepackage{makecell}
\usepackage{hyperref}
\usepackage{adjustbox}
\numberwithin{equation}{section}
\usepackage{makecell}
\usepackage{bm}
\usepackage{xcolor}
\newtheorem{theorem}{Theorem}
\newtheorem{proposition}{Proposition}
\newtheorem{corollary}{Corollary}
\newtheorem{lemma}{Lemma}
\newtheorem{definition}{Definition}

\newtheorem{example}{Example}

\usepackage{enumerate}
\usepackage{array}
\usepackage[nottoc,notlot,notlof]{tocbibind} 
\title{ Function-Correcting Codes for Linear and Locally Bounded Functions Over a Finite Chain Ring }
\author{Gyanendra K. Verma\footnote{ The author is financially supported by the ANRF-NPDF grant PDF/2025/002309. \\ email: gkvermaiitdmaths@gmail.com} \\ Department of Electrical Communication Engineering,\\ Indian Institute of Science Bangalore, Karnataka 560012, India.\\\\
 Abhay Kumar Singh\footnote{email: abhay@iitism.ac.in.}\\Department of Mathematics and Computing,\\ Indian Institute of Technology (ISM) Dhanbad, Jharkhand - 826004, India}
\date{}

\begin{document}
	\maketitle
\begin{abstract}
Redundancy is important in error correction because it adds extra information that enables the detection and correction of errors when data becomes corrupted. However, determining the optimal amount of redundancy is often a challenging problem in coding theory. To address this, a new coding framework called function-correcting codes has been introduced. Function-correcting codes are a significant development in which the codes are designed to protect only the function values of the message data rather than the entire message itself. This approach allows reliable recovery of the desired output while reducing the required redundancy, leading to more efficient error correction. Recently, a theoretical framework for systematic function-correcting codes (FCCs) designed for channels equipped with a homogeneous metric has been developed. 

In this paper, we further extend the study of function-correcting codes in the homogeneous metric over a chain ring $\mathbb{Z}_{2^s}$ for broader classes of functions, namely, locally bounded functions and linear functions, and for weight functions, modular sum functions.  
We define locally bounded functions in the homogeneous metric over $\mathbb{Z}_{2^s}^k$ and investigate the locality of weight functions. We derive a Plotkin-like bound for irregular homogeneous distance code over $\mathbb{Z}_4$, which improves the existing bound. Using locality properties of functions,  we establish upper and lower bounds on the optimal redundancy. We provide several explicit constructions of function-correcting codes for locally bounded functions, weight functions, and weight distribution functions.  Using these constructions, we further discuss the tightness of the derived bound. We explicitly derive a Plotkin-like bound for linear function-correcting codes that  reduces to the classical Plotkin
bound when the linear function is bijective, we further discuss a construction of
function-correcting linear codes over $\mathbb{Z}_{2^s}$.
\end{abstract}

\section{Introduction}
The problem of minimizing redundancy subject to reliability and decoding constraints constitutes a major challenge in coding theory. This challenge motivated Lenz et al. \cite{Lenz2023} to introduce a new coding framework, termed function-correcting codes, which aims to reduce redundancy in scenarios where only specific functions of the data, rather than the data itself, must be recovered. Designing efficient function-correcting codes requires careful consideration of the underlying communication channel (or metric) and the target functions, as these factors are essential for redundancy optimization and for deriving fundamental bounds. Every systematic error-correcting code is, in particular, a function-correcting code for arbitrary functions defined on the underlying alphabet set. Codes over rings are natural generalizations of classical linear codes over fields and often capture algebraic structure that field-based models fail to reflect, especially in the presence of constraints or nonstandard noise. They provide a unified framework for describing important families of nonlinear binary codes and for designing codes matched to modulation schemes (e.g., phase-modulated constellations) and memory or storage architectures. Moreover, studying codes over rings enriches the theory of duality, weight enumerators, and MacWilliams-type identities, leading to new structural insights and potentially more powerful coding schemes. In \cite{Wood2007}, Wood investigated code equivalence over finite rings and proved that any finite ring on which MacWilliams’ code equivalence theorem holds for the Hamming weight must necessarily be Frobenius. The homogeneous distance is typically introduced in an intrinsic algebraic or metric framework, without an explicit association to particular channel models in \cite{constantinescu1997}. In this work, we investigate function-correcting codes endowed with the homogeneous distance when transmission occurs over additive white Gaussian noise (AWGN) channels analogous to \cite{liu2025function}. Within the context of codes over finite rings, the homogeneous distance constitutes a more appropriate metric than the classical Hamming distance, as it intrinsically accounts for symbol errors of differing severity. Following the fundamental paper of Lenz et al. \cite{Lenz2023}, the framework of function-correcting codes (FCCs) was extended to the symbol-pair metric over binary fields \cite{Xia2023} and to the 
$b$-symbol metric over finite fields \cite{Singh2025}. Verma and Singh \cite{Verma2025FCCL} introduced the framework of function-correcting codes for the Lee metric. They derived bounds on optimal redundancy and provided constructions of function-correcting codes for several functions. Recently, in \cite{Pandey2026,Rajput2026-2} the authors introduced and studied function-correcting partition codes and function-correcting codes for maximally-unbalanced boolean functions, respectively. By exploiting the significance of the homogeneous metric, Liu \cite{liu2025function} studied the framework of FCCs under the homogeneous metric and established bounds on optimal redundancy. They also presented constructions of function-correcting codes for certain special classes of functions. 

It is evident from the seminal work of Lenz \emph{et al.}~\cite{Lenz2023} that the choice of the target function plays a crucial role in the design of function-correcting codes. In this context, locally bounded functions form an important class, as they enable the optimization of redundancy in such codes. In~\cite{Rajput2025}, the authors provided an explicit construction of function-correcting codes for locally bounded functions under the Hamming distance. They also proposed an upper bound on the redundancy of these codes, derived from the minimum possible length of an error-correcting code with a prescribed number of codewords and a given minimum Hamming distance.

A more general framework for function-correcting codes with locally bounded functions was introduced in~\cite{Verma2026} using the \(b\)-symbol metric. The authors showed that, under certain constraints, a function can be regarded as locally bounded. Furthermore, they investigated the locality and optimal redundancy of function-correcting \(b\)-symbol codes corresponding to the \(b\)-symbol weight function and the \(b\)-symbol weight distribution function for \(b \ge 1\). Table \ref{tabredundancy}, Table \ref{tab2}, and Table \ref{tab3} summarize the metrics considered and the corresponding classes of functions for which function-correcting codes have been studied.

 \begin{table*}[h!]
    \centering 
    \renewcommand{\arraystretch}{1.5}
    \resizebox{\textwidth}{!}{
    \begin{tabular}{|c|c|c|c c|}
    \Xhline{2\arrayrulewidth}
    \makecell{\textbf{Metric}} & \textbf{Locally $(\lambda, \rho)$-function} & \textbf{Optimal redundancy} & \textbf{References} & \\
    \Xhline{2\arrayrulewidth}
    \multirow{4}{*}{ $(\mathbb{F}_2^k,d_H)$} & Locally $(2, 2t)$-function & $r_f^H(k,t) = 2t$ & \cite[Lemma 5]{Lenz2023}&\\
    \cline{2-5}
                         & Locally $(3, 2t)$-function & $r_f^H(k,t) = 3t$ & \cite[Theorem 6]{Verma2026}& \\
    \cline{2-5}
                         & Locally $(4, 2t)$-function & $r_f^H(k,t) \leq  3t$ & \cite[Lemma 3]{Rajput2025}& \\ 
    \cline{2-5}
                         & Locally $(\lambda, 2t)$-function  & $r_f^H(k,t) \leq \lambda t$ & \cite[Theorem 5]{Verma2026}& \\                        
    \Xhline{3\arrayrulewidth}
    \multirow{2}{*}{ $(\mathbb{Z}_m^k,d_L)$, $m\geq 2$}  & Locally $(2, 2t)_L$-function & $r_f^L(k,t) = \left\lceil \frac{2t}{\left\lfloor \frac{m}{2} \right\rfloor} \right\rceil $ & \cite[Proposition 5.6]{Verma2025FCCL} & \\ 
    \cline{2-5}
                        & Locally $(\lambda, 2t)_L$-function & $r^L_f(k,t) \leq \left\lceil \frac{t}{\left\lfloor \frac{m}{2\lambda} \right\rfloor} \right\rceil $, $\lambda\leq \frac{m}{2}$ &\cite[Lemma 16]{rajan2025explicit} & \\
    \cline{2-5}
                        & Locally $(3, 2t)_L$-function & $r_f^L(k,t) = t,$ $m=6$ &\cite[Lemma 17]{rajan2025explicit} & \\
    \Xhline{3\arrayrulewidth}
    \multirow{4}{*}{$(\mathbb{F}_2^k,d_b)$}  & Locally $(2, 2t)$-function, $b=2$ & $2t - 2\leq r_f^b(k,t) \leq 2t - 1$ & \cite[Lemma 14]{Xia2023} &\\
    \cline{2-5}
                & Locally $(2, 2t)$-function, $b\geq 1$ & $ 2(t- b + 1) \leq r_f^b(k,t) \leq 2t - b + 1 $ & \cite[Lemma 5.1]{Singh2025} & \\         
    \cline{2-5}   
                & Locally $(4, 2t)$-function, $b\geq 1$ & $ r_f^b(k, t) \leq 3t-b+1$ & \cite[Lemma 10]{Verma2026} &  \\  

    \cline{2-5} 
                & Locally $(2^b, 2t)$-function, $b\geq 1$ & $2(t - b + 1) \leq r_f^b(k, t) \leq 2t$ & \cite[Corollary 15]{Verma2026}  & \\  
    \Xhline{3\arrayrulewidth}     
     \multirow{3}{*} {$(\mathbb{Z}_{2^s}^k,d_h)$}                & Locally $(2, 2t)_h$-function & $ r_f^h(k,t) =t$ & \cite[Theorem 5.1]{liu2025function} &\\
     \cline{2-5}
                & Locally $(4, 2t)_h$-function,  &
       $ r_f^h (k, t) \leq 2t$
   & Proposition \ref{lambda4} &\\ 
     \cline{2-5}
                & Locally $(\lambda, 2t)_h$-function,  &
       $ r_f^h (k, t) \leq \begin{cases}
            \frac{\lambda t}{2} & \text{if } t \text{ is even}\\
            \frac{\lambda (t+1)-2}{2} & \text{if } t \text{ is odd}.
        \end{cases}$
   & Corollary \ref{optimal-upperbound} &\\ 
    \Xhline{3\arrayrulewidth} 
    \end{tabular}}
    \caption{Bounds on optimal redundancy of function-correcting codes in various metric spaces for locally bounded functions.}\label{tabredundancy}
    \end{table*}

 \renewcommand{\arraystretch}{3}\begin{table}[h!]
        \centering
  \footnotesize \begin{tabular}{|c|c|c|}
        \hline
         Metric    & Plotkin-like bound  & References\\
         \hline
          $(\mathbb{F}_2^k,d_H)$   & $N(D)\geq \begin{cases}
              \frac{4}{M^2}\sum_{i,j:i<j} [D]_{ij}, \text{ if } M \text{ is even, }\\
               \frac{4}{M^2-1}\sum_{i,j:i<j} [D]_{ij}, \text{ if } M \text{ is odd }
          \end{cases}$ & \cite[Lemma 1]{Lenz2023} \\
          \hline
             $(\mathbb{F}_q^k,d_b)$ & $N_b(D)\geq \begin{cases}
              \frac{2q^b}{(q^b-1)M^2}\sum_{i,j:i<j} [D]_{ij}, \text{ if } M\equiv 0 \pmod{q^b} \\
               \frac{2q^b}{(q^b-1)(M^2-1)}\sum_{i,j:i<j} [D]_{ij}, \text{ if } M\equiv 0 \pmod{q^b}\\
               \frac{2q^b}{M^2(q^b-1)-m(q^b-m)}\sum_{i,j:i<j} [D]_{ij}, \text{ if } M \equiv m \pmod{q^b}
          \end{cases}$  & \cite[Lemma 4.1]{Singh2025} \\
          \hline
           $(\mathbb{Z}_m^k,d_L)$   & $N_L(D)\geq \begin{cases}
              \frac{8}{M^2m}\sum_{i,j:i<j} [D]_{ij}, \text{ if } m \text{ is even, }\\
               \frac{8m}{M^2(m^2-1}\sum_{i,j:i<j} [D]_{ij}, \text{ if } m \text{ is odd }
          \end{cases}$ & \cite[Lemma 4.1]{Verma2025FCCL}  \\
             \hline
        $(\mathbb{Z}_{2^s}^k,d_h)$   &$N_h(D)\geq \frac{1}{M^2} \sum_{i,j} [D]_{ij}$  & \cite[Lemma 3.2]{liu2025function} \\
           \hline
           $(\mathbb{Z}_{4}^k,d_h)$   & $N_h(D)\geq \begin{cases}
                \frac{1}{M^2} \sum_{i,j}[D]_{ij}, & \text{ if } M\equiv 0,2 \pmod{4}\\
            \frac{1}{M^2-1}\sum_{i,j}[D]_{ij}, & \text{ if } M\equiv 1,3 \pmod{4}
           \end{cases}$  & \makecell{Theorem \ref{plotkinimp4}} \\
           \hline
        \end{tabular}
        \caption{Plotkin-like bounds on irregular distance codes with respect to various metrics}
        \label{tab2}
    \end{table}

  \renewcommand{\arraystretch}{1.5}  \begin{table}[h!]
        \centering
        \begin{tabular}{|c|c|c|}
        \hline
         Metric    &  Plotkin-like bound for linear function & References\\
         \hline
           $(\mathbb{F}_2^k,d_H)$   & $r_f^H(k,t)\geq \frac{q}{q-1}(2t+1)(1-q^{-l})-k+\frac{s}{(q-1)q^{k-1}}$ & \cite[Corollary 5]{Premlal2024} \\
           \hline
             $(\mathbb{F}_2^k,d_b)$ & $r_f^b(k,t)\geq \frac{q^b}{q^b-1}(2t-b+2)(1-q^{-l})-k+\left (\frac{q^b}{q^b-1}\right) \left (\frac{s}{q^k}\right)$  & \cite[Theorem 6]{Sampath2025} \\
           \hline
           $(\mathbb{Z}_m^k,d_L)$   & $r_f^L(k,t)\geq \frac{m}{S}(2t+1)(1-m^{-\ell})-k+\frac{s}{Sm^{k-1}}$ & \cite[Theorem 7.1]{Verma2025FCCL}  \\
             \hline
           $(\mathbb{Z}_{2^s}^k,d_h)$   & $r_f^h(k,t)\geq (2t+1)(1-2^{-s\ell})-k+\frac{A}{2^{sk}}$  & \makecell{ Theorem \ref{plotlinear}} \\
           \hline
        \end{tabular}
        \caption{Plotkin-like bounds on function-correcting codes for linear functions with respect to various metrics}
        \label{tab3}
    \end{table}
\noindent    
\textbf{ Our contributions}: In this work, we study function-correcting codes in the homogeneous metric over $\mathbb{Z}_{2^s}$ for several functions, including the weight function and its associated distribution function, the modular sum function, and broader classes for locally bounded functions and linear functions. The key technical contributions of this paper are as follows: \\
$\bullet$ We define locally bounded functions over $\mathbb{Z}_{2^s}^k$ with respect to a homogeneous metric and investigate locality properties of the homogeneous weight function and the homogeneous weight distribution function. (Section \ref{sectionlocally}) \\
$\bullet$ We provide a Poltkin-like bound for irregular homogeneous distance codes over $\mathbb{Z}_4$, that improves the existing bound. (Theorem \ref{plotkinimp4})\\
$\bullet$  We establish several upper and lower bounds on the optimal redundancy of function-correcting codes with homogeneous distance for locally bounded functions, and provide explicit constructions of function-correcting codes. (Section \ref{sectionredundancy})\\
$\bullet$  We investigate linear function-correcting codes and establish a Plotkin-like bound. Finally, we discuss the construction of function-correcting linear codes over $\mathbb{Z}_{2^s}$. (Section \ref{FCC-linear})\\\\

\noindent
\textbf{Organization}: Section \ref{section-pre} introduces the basic definitions and notation used in this paper and recalls fundamental results related to the homogeneous metric and function-correcting codes. In Section \ref{sectionlocally}, we introduce the notion of locally bounded functions under the homogeneous metric over $\mathbb{Z}_{2^s}$ and analyze the locality properties of the homogeneous weight function and its associated weight-distribution functions.
Section \ref{sectionredundancy} establishes lower and upper bounds on the optimal redundancy for locally bounded functions and presents explicit code constructions. In Section \ref{FCC-linear}, we investigate function-correcting codes for linear functions in the homogeneous metric over $\mathbb{Z}_{2^s}$, deriving a Plotkin-type bound for these functions that reduces to the classical Plotkin bound for error-correcting codes when the underlying linear function is bijective. We also discuss constructions of function-correcting linear codes over $\mathbb{Z}_{2^s}$. Finally,  Section \ref{sectionconclusion} summarizes the paper and outlines possible directions for future work.

\textbf{Notations:}
\begin{align*} 
    \text{FCC}: & \text{ function-correcting code}\\
    \text{FCCHD}: & \text{ function-correcting code with homogeneous distance}\\
    [M]: & \text{ set of positive integers } \{1,2,\dots, M\}\\
    \mathbb{N}_0: & \text{ set of all non-negative integers}\\
    \mathbb{F}_q:    & \text{ finite field with } q\text{ elements, where } q \text{ is a prime power}  \\
        \mathbb{Z}_m: &   \text{ ring of integers modulo } m\\
     d_H: &  \text{  Hamming metric }  \\
     d_b: & \ b  \text{-symbol metric  }  \\
     d_L: &  \text{ Lee metric}  \\
     d_h: &  \text{ homogeneous metric}  \\
     w_h : & \text{ homogeneous weight function}\\
     r_f^{H} (k,t): & \text{ the optimal redundancy of an } (f,t) \text{- FCC in the Hamming metric }\\
     r_f^{b} (k,t): & \text{ the optimal redundancy of an } (f,t) \text{- FCC in the $b$-symbol metric }\\
     r_f^{L} (k,t): & \text{ the optimal redundancy of an } (f,t) \text{- FCC in the Lee metric }\\
     r_f^{h} (k,t): & \text{ the optimal redundancy of an } (f,t) \text{- FCC in the homogeneous metric }\\
     % \lfloor x \rfloor : & \text{ floor function value of a real number } x\\
     (\lambda,\rho)_X\text{-function} : & \ \rho\text{-locally } \lambda \text{-bounded function with respect to metric $d_X$}, X= H, L, h.\\
     \text{Im}(f): & \text{ image set of function} f\\
     N_h(\lambda, 2t): & \text{ the minimum length of a code with } \lambda \text{ codewords} \\  & \text{ and  minimum homogeneous distance } 2t
\end{align*}

    % \subsection{Motivation}

    % \subsection{Related Work}

    % \subsection{Contributions}

    % \subsection{Organization}
    % \subsubsection{Notations}
    
\section{Preliminaries}\label{section-pre}

Let $\mathbb{Z}_{2^s}=\{0,1,\dots,2^s-1\}$ denote the ring of integers modulo $2^s$, is a finite commutative local ring whose ideals are totally ordered by inclusion, that is, a finite chain ring with unique maximal ideal generated by $2$. In this section, we introduce the fundamental definitions and preliminary results pertaining to function-correcting codes and the homogeneous metric over $\mathbb{Z}_{2^s}$. We begin by recalling the definition of homogeneous weight on an arbitrary finite ring.
 
\begin{definition}\cite[Homogeneous Weight]{Greferath2000} 
Let $R$ be a finite ring. A function $w_h$ on $R$ is called a (left) homogeneous weight if $w_h(0)=0$ and $w_h$ satisfies the following hold 
\begin{enumerate}
    \item For $a,b \in R$, $w_h(a)=w_h(b)$ whenever $Ra=Rb$.
    \item There exists a real number $\alpha>0$ such that 
    $$\sum_{b\in Ra}w_h(b)=\alpha |Ra|\ \ \  \text{ for all } a\in R\setminus\{0\},$$
\end{enumerate}
where $Ra=\{z\cdot a| z\in R\}$ is the left ideal generated by $x$.\\
The homogeneous distance between $a,b\in R$ is $d_h(a,b)=w_h(a-b)$.
\end{definition}
By \cite{Greferath2000}, the homogeneous weights are uniquely determined up to a scalar multiple of $\alpha$. Therefore, without loss of generality, we can assume $\alpha=1$. Using the formula of homogeneous weight given in \cite{Fan2013,Liu2020}, we have a formula to calculate the homogeneous weight over $\mathbb{Z}_{2^s}$ (see \cite{liu2025function}). Let $a\in \mathbb{Z}_{2^s}$. Then,
\begin{align}\label{hweight}
    w_h(a)=\begin{cases}
        0 & \text{if } a=0\\
        1 & \text{if } a\notin \langle 2^{s-1}\rangle\\
        2 & \text{if } a\in \langle 2^{s-1}\rangle\setminus\{0\}.
    \end{cases}
\end{align}

For $x=(x_1,x_2,\dots,x_k)$ and $y=(y_1,y_2,\dots,y_k)$ in $\mathbb{Z}_{2^s}^k$, $$d_h(x,y)=w_h(x-y)=\sum_{i=1}^kw_h(x_i-y_i).$$ In general, the homogeneous distance may not be a metric over $\mathbb{Z}_{2^s}$. From \cite[Theorem 2]{constantinescu1997}, the homogeneous distance over $\mathbb{Z}_m$ defines a metric if and only if $m\not \equiv 0 \pmod{6}$. Thus, the homogeneous distance is a metric over $\mathbb{Z}_{2^s}$ for all positive integers $s$. Moreover, for $s=1$, the homogeneous distance coincides with the Hamming distance over $\mathbb{Z}_2$ and for $s=2$, it coincides with the Lee metric over $\mathbb{Z}_4$. 

\begin{definition}
The minimum homogeneous distance of a code $C$ is defined as
\begin{equation*}
d_h(C) = \in\{d_h(x,y)| x,y\in \mathbb{Z}_{2^s}, x\neq y\}. 
\end{equation*}
 \end{definition}

Let $x\in\mathbb{Z}_{2^s}^k$. Define a ball $B_h(u,\rho)$  with center  $x$ and radius $\rho$ as follows
$$B_h(x,\rho)=\{y\in\mathbb{Z}_{2^s}^k| \ d_h(x,y)\leq \rho\}.$$

Function-correcting codes were introduced by Lenz et al. \cite{Lenz2023} in the context of the Hamming metric. Later, in \cite{liu2025function}, the authors studied function correcting codes for the homogeneous distance. Here, we briefly recall the basic definitions and results that are used throughout the paper. 

Let $f:\mathbb{Z}_{2^s}^k\to \text{Im}(f)$ be a function. Define an encoding 
\begin{align*}
    Enc:\mathbb{Z}_{2^s}^k\to \mathbb{Z}_{2^s}^{k+r}, \ Enc(x)=(x,x_p),
\end{align*}
where $x_p\in \mathbb{Z}_{2^s}^r$ is called redundancy vector. The set $C=\{Enc(x)| x\in \mathbb{Z}_{2^s}^k\}$ is called a code of length $k+r$ (equivalently, with redundancy $r$) over $\mathbb{Z}_{2^s}^k$. 

\begin{definition}\label{fcchd}[Function-correcting code with homogeneous distance]
Let $t$ be a positive integer and $f:\mathbb{Z}_{2^s}^k\to \text{Im}(f)$ be a function. An encoding function
\begin{align*}
\text{Enc}: \mathbb{F}_2^k \rightarrow \mathbb{F}_2^{k+r}, \quad \text{Enc}(x) = (x, x_p)\  \forall x\in \mathbb{F}_2^k
\end{align*}
is said to define an $(f,t)$-function-correcting code with homogeneous distance (in short FCCHD)

if for all $x, y \in \mathbb{Z}_{2^s}^k$ with $f(x) \neq f(y)$, the following holds
\begin{align*}
d_h(\text{Enc}(x), \text{Enc}(y)) \geq 2t + 1.
\end{align*}
\end{definition}
The value $r$ is referred to as the redundancy of function-correcting codes, and $k+r$ is called the code length. By definition, if at most $t$ transmission errors occur, the receiver can still reconstruct the correct value $f(u)$ using function $f$ and encoding $Enc$. Note that an error-correcting code  $C$ in the homogeneous metric with minimum homogeneous distance $2t+1$ can correct up to $t$ errors.  A code of length $k+r$ is said to be a systematic code if its first $k$ coordinates correspond to the message symbols and the remaining $r$ coordinates correspond to parity symbols.  Therefore, any systematic $t$-error-correcting code in the homogeneous metric yields, in a straightforward manner, an $(f,t)$-FCCHD.

\begin{definition}[Optimal redundancy]
Let $t>0$ be a integer and $f:\mathbb{Z}_{2^s}^k\to \text{Im}(f)$ be a function.
The optimal redundancy $r_f^h(k,t)$ of an $(f,t)$-FCCHD is the smallest integer $r>0$ for which there exists an $(f,t)$-FCCHD with redundancy $r$. 
\end{definition}

\begin{definition}[Homogeneous distance requirement matrix]
Let $f:\mathbb{Z}_{2^s}^k \to \text{Im}(f)$ be a function. Consider $M$ vectors $x_1, \ldots, x_{M} \in \mathbb{Z}_{2^s}^k$. Then, the homogeneous distance requirement matrix $D_f^{h}(t, x_1, \ldots, x_{M})$  defined as 
\begin{align*}
  [D_f^{h}(t, x_1, \ldots, x_{M})]_{ij} =
\begin{cases}
[2t +1 - d_h(x_i, x_j)]^+ & \text{if } f(x_i) \neq f(x_j), \\
0 & \text{otherwise},
\end{cases}  
\end{align*}
\end{definition}

\begin{definition}[$D_h$-Code]
Let $D$ be a $M\times M$ matrix over $\mathbb{N}_0$ with $[D]_{ii}=0$ for all $i\in [M]$. Let $P = \{ p_1,p_2, \ldots, p_M \}$ be a code of length $r$ over $\mathbb{Z}_{2^s}$. Then $P$ is called a $D$-homogeneous distance code (in short $D_h$-code) if there exists an ordering of codewords in $P$ such that  
$$d_h(p_i,p_j) \geq [D]_{ij} \quad \text{for all } i, j \in [M].$$
\end{definition}
Let $N_h(D)$ denote the smallest positive integer $r$ for which there exists a $D$- homogeneous distance code of length $r$. 

% The following lemma gives a  bound on the length of a irregular homogeneous distance code. 
% \begin{lemma}\cite[Lemma 3.2]{liu2025function}\label{plotkin}
%   Let $D\in \mathbb{N}_0^{M\times M}$ be a matrix. Then
%     \begin{align*}
%         N_h(D)\geq \frac{1}{M^2}\sum_{i,j}[D]_{ij}.
%     \end{align*}
% \end{lemma}

% \begin{definition}\cite{liu2025function}
% Let $f:\mathbb{Z}_{2^s}^k\to \text{Im}(f)$ be a function. For $e\in \text{Im}(f)$, let $f^{-1}(e)=\{x\in \mathbb{Z}_{2^s}^k| f(x)=e\}$. For $e_1,e_2$, the function homogeneous distance is defined as   
% \begin{align*}
%     d_h(e_1,e_2)=\min\{d_h(x,y)| x\in f^{-1}(e_1), y\in f^{-1}(e_2)\}.
% \end{align*}
% \end{definition}

% \begin{definition}\cite{liu2025function}
    
% \end{definition}

The following lemma gives a lower bound on the length of an irregular distance code with a homogeneous metric over $\mathbb{Z}_{2^s}$. 
\begin{lemma}\cite[Lemma 3.2]{liu2025function}\label{plotkin}
  Let $D\in \mathbb{N}_0^{M\times M}$ be a matrix. Then
    \begin{align*}
        N_h(D)\geq \frac{1}{M^2}\sum_{i,j\in [M]}[D]_{ij}.
    \end{align*}
\end{lemma}

\begin{corollary}\cite[Corollary 3.1]{liu2025function}\label{nbnh}
   Let $x_1,x_2,\dots,x_M$ be distinct vectors in $\mathbb{Z}_{2^s}^k$. Then the optimal redundancy of an $(f,t)$-FCCHD for any function $f:\mathbb{Z}_{2^s}^k\to \text{Im}(f)$ satisfies
   \begin{align*}
       r_f^h(k,t)\geq N_h(D_f^h(t,x_1,x_2,\dots,x_M)).
   \end{align*}
Moreover, if $s\geq 2$ and $|\text{Im}(f)|\geq 2$, then, $r_f^h(k,t)\geq t$. 
\end{corollary}

\section{Locally bounded functions in the homogeneous metric}\label{sectionlocally}
We define locally bounded functions in the homogeneous metric over $\mathbb{Z}_{2^s}^k$ and investigate the locality of the homogeneous weight and the homogeneous weight distribution functions. Let $x\in \mathbb{Z}_{2^s}^k$. Define a ball $B_{h}(x,\rho)$  with center  $x$ and radius $\rho$ as follows
$$B_{h}(x,\rho)=\{y\in \mathbb{Z}_{2^s}^k| \ d_h(x,y)\leq \rho\}.$$

\begin{definition}[Function homogeneous Ball]
Let  $f : \mathbb{Z}_{2^s}^k \to \operatorname{Im}(f)$ be a function. For a vector  $x \in \mathbb{Z}_{2^s}^k$ and a non-negative integer $\rho$, the function homogeneous ball of  $f$  with radius  $\rho$ centered at $x$ is defined as
\begin{align*}
   B_f^h(x, \rho) = \{ f(y) \mid y \in \mathbb{Z}_{2^s}^k \text{ and } d_h(x, y) \leq \rho \}. 
\end{align*}
\end{definition}
In \cite{liu2025function}, the authors defined locally binary homogeneous functions and investigated the optimal redundancy for these functions. Next, we define locally bounded homogeneous functions, which generalize the locally binary homogeneous functions.
\begin{definition}\label{deflocally}
 A function $f:\mathbb{Z}_{2^s}^k\to \text{Im}(f)$ is said to be a  $\rho$-locally $\lambda$-bounded homogeneous function (in short $(\lambda, \rho)_h$-function) if \begin{align*}
     |B_f^h(x,\rho)|\leq \lambda \text{ for all } x\in \mathbb{Z}_{2^s}^k
 \end{align*}
\end{definition}
For $\lambda=2$, the above definition reduces to locally binary homogeneous functions. In subsequent sections, we demonstrate that the value of $\lambda$ is pivotal in deriving the optimal redundancy bounds for FCCHD associated with a given function. Therefore, determining the smallest value of $\lambda$ for which a given function $f$ is a $(\lambda,\rho)_h$-function, for a given value of $\rho$, is significantly important. Observe that any function is $(\lambda,\rho)_h$-function when $\lambda = |\text{Im}(f)|$, irrespective of the choices of $\rho$. Similarly, when $\rho = 2k$, the function $f$ is  $(\lambda,\rho)_h$-function only if $\lambda \geq |\text{Im}(f)|$. Since these scenarios are relatively straightforward and offer limited theoretical insights, our attention is primarily directed toward the study of $(\lambda,\rho)_h$-functions under the more restrictive conditions $\lambda \leq |\text{Im}(f)|$ and $\rho \leq 2k$.  We denote $\lambda_0$ to be the smallest value of $\lambda$ for which a function $f$ is a $(\lambda,\rho)_h$-function. 
\begin{corollary}\label{balllocal}
   Let $f:\mathbb{Z}_{2^s}^k\to \text{Im}(f)$ be a $(\lambda_0,\rho)_h$-function. Then $\lambda_0\leq |B_h(\mathbf{0},\rho)|$. 
\end{corollary}

\subsection{Locality of the weight functions}
In this subsection, we explore the locally boundedness of the homogeneous weight and associated weight distribution function. In particular, we determine admissible values of $\lambda$ for which these functions are $(\lambda,\rho)_h$-functions.

\begin{definition}
    The homogeneous weight function is defined $w_h:\mathbb{Z}_{2^s}^k\to \{0,1,\dots, 2k\}$, $w_h(x)=$ homegeneous weight of $x\in\mathbb{Z}_{2^s}^k $. For $T>0$, the homogeneous weight distribution function is defined as
    $$\Delta_T^h(x)= \left \lfloor \frac{w_h(x)}{T} \right \rfloor, \forall x\in \mathbb{Z}_{2^s}^k .$$
\end{definition}

\begin{proposition}\label{proplocalhwt}
 Let $\rho$ be a positive integer. Then the homogeneous weight function is a locally $(2\rho+1,\rho)_h$-function.    
\end{proposition}
\begin{proof}
Let $x\in \mathbb{Z}_{2^s}^k$ and $y\in B_h(x,\rho)$. Then, by the triangle inequality, we have
\begin{align*}
    \begin{split}
        w_h(y)=d_h(y,0)\leq d_h(y,x)+d_h(x,0)\leq \rho+w_h(x)
    \end{split}
\end{align*}
and 
\begin{align*}
    \begin{split}
        w_h(x)=d_h(x,0)\leq d_h(x,y)+d_h(y,0)\leq \rho+w_h(y).
    \end{split}
\end{align*}
Therefore, for any $y\in B_h(x,\rho)$, we have 
\begin{equation} \label{hwt}
    w_h(x)-\rho \leq w_h(y)\leq w_h(x)+\rho.
\end{equation} 
Hence, $|B_{w_h}^h(x,\rho)|\leq 2\rho+1$ for all $x\in \mathbb{Z}_{2^s}^k$.

\end{proof}

\begin{proposition}\label{lamhwd}
 The homogeneous weight distribution function $\Delta_T^h$ with threshold $T$ is a locally $(\left\lfloor\frac{2\rho}{T}\right\rfloor+2,\rho)_h$-function. 
\end{proposition}
\begin{proof}
  Let $x\in \mathbb{Z}_{2^s}^k$.  From the proof of Proposition \ref{proplocalhwt}, for $y\in B_h(x,\rho)$, we have 
    \begin{equation*} 
    \left \lfloor \frac{w_h(x)-\rho}{T}\right \rfloor \leq \Delta_T^h(y)\leq \left \lfloor \frac{w_h(x)+\rho}{T}\right \rfloor.
\end{equation*} 

Therefore, $$\mid B_{\Delta_T^h}^h(x,\rho)\mid= \left \lfloor \frac{w_h(x)+\rho}{T}\right \rfloor-\left \lfloor \frac{w_h(x)-\rho}{T}\right \rfloor+1\leq \left \lfloor \frac{(w_h(x)+\rho)- (w_h(x)-\rho)}{T}\right \rfloor+2.$$
 As $x$ is arbitrary, therefore, $\Delta_T^h$ is a locally $\left (\left \lfloor\frac{2\rho}{T}\right \rfloor + 2,\rho\right )_h$-function.
     This completes the proof.
\end{proof}

\begin{corollary}
     Let $k,\  \rho$ be positive integers and $\rho$ is even with $\rho<2k $. Then the Hamming weight distribution function $\Delta_T^h$ with threshold $T=2\rho+1$, is a locally $\left(\lambda_0=\left\lfloor\frac{2\rho}{T}\right\rfloor+2,\rho\right)_h$-function. 
\end{corollary}
\begin{proof}
Let $x=(\underbrace{2^{s-1},2^{s-1},\dots,2^{s-1}}_{\frac{\rho}{2}-times}, 1,0,0,\dots,0)\in \mathbb{Z}_{2^s}^k$. Then $w_h(x)=\rho+1$.
Similar to the proof of Proposition \ref{lamhwd}, we have

$$\mid B_{\Delta_T^h}^h(x,\rho)\mid= \left \lfloor \frac{w_h(x)+\rho}{T}\right \rfloor-\left \lfloor \frac{w_h(y)-\rho}{T}\right \rfloor+1= \left \lfloor \frac{(w_h(x)+\rho)- (w_h(x)-\rho)}{T}\right \rfloor+2=\left \lfloor \frac{2\rho}{T}\right\rfloor+2.$$
This completes the proof.
\end{proof}

    \begin{corollary}\label{rho+1}
        Let $k, \rho, T \in \mathbb{N}$ with $\rho \leq 2k$. If $T$ divides $\rho$, then  homogeneous weight distribution function $\Delta_T^h$ with threshold $T$, is a locally $\left(\lambda=\left\lfloor\frac{2\rho}{T}\right\rfloor+1,\rho\right)_h$-function. 
    \end{corollary}   
    \begin{proof}
Let $x\in \mathbb{Z}_{2^s}^k$ with $w_h(x)=w$. If $T\mid \rho$, then fractional parts $\left\{ \frac{w + \rho}{T}\right\} \geq \left\{ \frac{w - \rho}{T}\right\}$.  Consequently,
\begin{align*}
   |B_{\Delta_T^h}(x,\rho)|=\left\lfloor \frac{w + \rho}{T} \right\rfloor - \left\lfloor \frac{w - \rho}{T} \right \rfloor + 1  =  \left\lfloor \frac{w + \rho}{T} -  \frac{w - \rho}{T} \right \rfloor + 1 =  \left\lfloor \frac{2\rho}{T} \right \rfloor + 1.
\end{align*}
    \end{proof}

\begin{proposition}\label{ballbdd}
   Let  $\Delta_T^h$ be a $(\lambda_0,\rho)_h$-function for some $0\leq \rho \leq 2k$ and positive integer $T$. Then  $\lambda_0\geq\left\lfloor\frac{\rho} {T}\right\rfloor + 1$.
\end{proposition}
\begin{proof}
    Note that $$|B_h(\mathbf{0},\rho)|=\sum_{i=0}^\rho|S_h(\mathbf{0},i)|,$$
    where $|S_h(\mathbf{0},i)|=|\{x\in \mathbb{Z}_{2^s}^k|\ d_h(x,\mathbf{0})=w_h(x)=i\}|$ for $0\leq i\leq \rho$. Also, 
    $$\Delta_T^h(S_h(\mathbf{0},i))= \left \{\left\lfloor\frac{i}{T}\right\rfloor\right\}.$$ Thus $\max_{x\in \mathbb{Z}_{2^s}^k} |B^b_{\Delta_T^h}(x,\rho)|\geq \left\lfloor \frac{\rho}{T}\right\rfloor + 1 $. This completes the proof.
\end{proof}
Combining the preceding results, we obtain the following corollaries.
\begin{corollary}\label{luboundwd}
    Let the homogeneous weight distribution function $\Delta^h_T$  be a locally $(\lambda_0,\rho)_h$-function. Then   $\left\lfloor\frac{\rho}{T}\right\rfloor + 1\leq \lambda_0\leq \left\lfloor\frac{2\rho}{T}\right\rfloor+2$.
\end{corollary}

\begin{corollary}\label{luboundw}
    Let the homogeneous weight function be a locally $(\lambda_0,\rho)_h$-function. Then $\rho+1\leq \lambda_0\leq 2\rho+2$. 
\end{corollary}

\section{The Optimal Redundancy of FCCHD for locally bounded functions}\label{sectionredundancy}

In this section, we focus on the construction of function-correcting codes and their optimal redundancy in the homogeneous metric over $\mathbb{Z}_{2^s}$.   The Plotkin-like bound gives an important bound on the optimal redundancy of function-correcting codes. In the next theorem, we establish a Plotkin-like bound over $\mathbb{Z}_4$ which improves Lemma \ref{plotkin}.
\begin{theorem}\label{plotkinimp4}
  Let $s=2$, i.e. $Z_{2^s}=\mathbb{Z}_4$, and  $D\in \mathbb{N}_0^{M\times M}$ be a matrix. Then
    \begin{align*}
        N_h(D)\geq \begin{cases}
            \frac{1}{M^2} \sum_{i,j}[D]_{ij}, & \text{ if } M\equiv 0,2 \pmod{4}\\
            \frac{1}{M^2-1}\sum_{i,j}[D]_{ij}, & \text{ if } M\equiv 1,3 \pmod{4}.
        \end{cases}
    \end{align*}
\end{theorem}
\begin{proof}
    Let $C$ be a $D_h$-code of length $r$ and size $M$ over $\mathbb{Z}_4$. Observe that 
    \begin{align}\label{eqplot}
        \sum_{i,j} [D]_{ij}\leq \sum_{x,y\in C}d_h(x,y)=\sum_{x,y\in C}\sum_{j=1}^{r}d_h(x_j,y_j)=\sum_{j=1}^{r}\sum_{x,y\in C}d_h(x_j,y_j).
    \end{align}
  Write the codewords of $C$ as rows of a matrix $B$ of order $M\times r$.  Let $B_j$ denotes the $j$-th column of $B$. Let $n_a$ be number of $a\in \mathbb{Z}_{4}$ in column $B_j$. Then
  \begin{align*}
      \sum_{x,y\in C}d_h(x_j,y_j)=\sum_{a,b\in \mathbb{Z}_{4}}n_an_bd_h(a,b)=2(n_0n_1+2n_0n_2+n_0n_3+n_1n_2+2n_1n_3+n_2n_3).
\end{align*}

  The contribution of the $j$-th coordinate to the total pairwise distance is maximized when the elements of $\mathbb{Z}_{4}$ are distributed as uniformly as possible. Writing $M=4Q+e$ with $0\leq e<4$, the most uniform distribution consists of $e$ elements occurring $Q+1$ times and the remaining 
$q-e$ elements occurring $Q$ times.\\
\textbf{Case 1} Let $e=0$. Then $n_0=n_1=n_2=n_3=\frac{M}{4}$, and 
\begin{align*}
\sum_{a,b\in\mathbb{Z}_4}n_an_bd_h(a,b)\leq 2\left (\frac{M^2}{16}+2\frac{M^2}{16}+\frac{M^2}{16}+\frac{M^2}{16}+2\frac{M^2}{16}+\frac{M^2}{16}\right)=M^2.
\end{align*}
\textbf{Case 2} Let $e=1$. Then one out of $n_0,n_1,n_2,n_3$ is $Q+1=\frac{M-1}{4}+1$ and others are $Q=\frac{M-1}{4}$. WLOG, let $n_0=Q+1$, $n_1=n_2=n_3=Q$  Thus,
\begin{align*}
\sum_{a,b\in\mathbb{Z}_4}n_an_bd_h(a,b)\leq& 2 [(Q+1)Q+2(Q+1)Q+(Q+1)Q+Q^2+2Q^2+Q^2]\\
=&2[4(Q+1)Q+4Q^2]=M^2-1.
\end{align*}

\textbf{Case 3} Let $e=2$. Then two out of $n_0,n_1,n_2,n_3$ is $Q+1=\frac{M-2}{4}+1$ and other two are $Q=\frac{M-2}{4}$. WLOG, let $n_0=n_2=Q+1$, $n_1=n_3=Q$  Thus,
\begin{align*}
\sum_{a,b\in\mathbb{Z}_4}n_an_bd_h(a,b)\leq& 2 [(Q+1)Q+2(Q+1)^2+(Q+1)Q+Q(Q+1)+2Q^2+(Q+1)Q]\\
=&2[2(Q+1)^2+4(Q+1)Q+2Q^2]=M^2.
\end{align*}

\textbf{Case 4} Let $e=3$. Then three out of $n_0,n_1,n_2,n_3$ is $Q+1=\frac{M-3}{4}+1$ and other one is $Q=\frac{M-3}{4}$. WLOG, let $n_0=n_1=n_2=Q+1$, $n_3=Q$  Thus,
\begin{align*}
\sum_{a,b\in\mathbb{Z}_4}n_an_bd_h(a,b)\leq& 2 [(Q+1)^2+2(Q+1)^2+(Q+1)Q+(Q+1)^2+2(Q+1)Q+(Q+1)Q]\\
=&2[4(Q+1)^2+4(Q+1)Q]=M^2-1.
\end{align*}
By Eq. \ref{eqplot}, 
\begin{align*}
    \sum_{i,j}[D]_{ij}\leq\begin{cases}
        rM^2 & \text{ if }M\equiv 0,2 \pmod{4}  \\
        r(M^2-1) & \text{ if } M\equiv 1,3 \pmod{4}.
    \end{cases} 
\end{align*}
The above holds for any $D_h$-code of length $r$. Thus $N_h(D)\geq r$. This completes the proof.
\end{proof}

Let $P$ be a set and $\prec$ be a binary relation on $P$. Then $\prec$ is called a total order on $P$ if it satisfies (1) if $a\prec b$ and $b \prec a$, then $a=b$, (2) if $a\prec b$ and $b\prec c$, then $a\prec c$, (3) for any $a,b\in P$, either $a\prec b$, $b\prec a$, or $a=b$. We say $(P,\prec)$ is a total order set. A subset $I$ of a total order set $(P,\prec)$ is said to be a contiguous block if for all $a,b\in I$ with $a\prec b$ and for any $c\in P$ with  $a\prec c\prec b$, we have $c\in I$.

% \begin{lemma}\cite[Lemma 1]{Rajput2025}\label{lemma1frm11}
% Let $f:=\mathbb{F}_2^k\to \text{Im}(f)$ be a locally $(\lambda,\rho)$-function in the Hamming metric. Suppose there exists a total order $\prec$ on Im$(f)$ such that for every $x\in \mathbb{Z}_2^k$, the set $B_f(x,\rho)$  forms a contiguous block of consecutive elements with respect to $\prec$.  Then there exists a function $\tau_f:\mathbb{F}_2^k\to \text{Im}(f)$ such that for any $u,v\in \mathbb{F}_2^k$, if $f(u)\neq f(v)$ and $d_H(u,v)\leq \rho$,  then $\tau_f(u)\neq \tau_f(v)$.   
% \end{lemma}

The following lemma is fundamental to the construction of function-correcting codes. Its proof is obtained by mirroring the proof given in  \cite[Lemma 1]{Rajput2025} for the Hamming distance. Thus, we omit the proof.

\begin{lemma}\label{tau}
Let $f:\mathbb{Z}_{2^s}^k\to \text{Im}(f)$ be a locally $(\lambda,\rho)_h$-function. Suppose there exists a total order $\prec$ on Im$(f)$ such that for every $x\in \mathbb{Z}_{2^s}^k$, the set $B_f^h(x,\rho)$  forms a contiguous block of consecutive elements with respect to $\prec$. Then there exists a map $\tau_f^h:\mathbb{Z}_{2^s}^k\to [\lambda]$ such that for all $x,y\in\mathbb{Z}_{2^s}^k$ with $f(x)\neq f(y)$ and $d_ h(x,y)\leq \rho$, we have $\tau_f^h(x)\neq \tau_f^h(y)$.
\end{lemma}
It is easy to see that the contiguous block condition in the above lemma always holds for the homogeneous weight and the homogeneous weight distribution function. Next, we construct several FCCHD and establish bounds on the optimal redundancy for locally bounded functions.

\begin{proposition}\label{lambda4}
 Let $t>0$ be an integer and $f:\mathbb{Z}_{2^s}^k\to \text{Im}(f)$ be a locally $(4,2t)_h$-function satisfying contiguous block condition in Lemma \ref{tau}. Then, there exists an FCCCHD with redundancy $r_f^h(k,t) = 2t$.
\end{proposition}
\begin{proof}
 By Lemma \ref{tau}, there exists a map $\tau_f^h:\mathbb{Z}_{2^s}^k\to [4]$ such that $\tau_f^h(x)\neq\tau_f^h(y)$ for all $x,y\in \mathbb{Z}_{2^s}^k$ whenever $ f(x)\neq f(y)$ and $d_h(x,y)\leq 2t$. For $x\in \mathbb{Z}_{2^s}^k$, define parity bits $x_p'$ corresponding to $x$ as follows
$$x_p'= \begin{cases}
00\hspace{1 cm}\text{ if } \hspace{0.5 cm}\tau_f^h(x)=1\\
0 a\hspace{1 cm}\text{ if } \hspace{0.5 cm}\tau_f^h(x)=2\\
a0\hspace{1 cm}\text{ if } \hspace{0.5 cm}\tau_f^h(x)=3\\
aa\hspace{1 cm}\text{ if } \hspace{0.5 cm}\tau_f^h(x)=4,
 \end{cases} \text{ where } a=2^{s-1}.$$
Let  $x_p=(x_p')^t\in \mathbb{Z}_{2^s}^{2t}$, the $t$ fold repetition of $x_p'$. Now, we show that $x_p$ define an FCCHD via following encoding map $Enc:\mathbb{Z}_{2^s}^k\to \mathbb{Z}_{2^s}^{k+2t}$ defined as $Enc(x)=(x,x_p)$.
For $x,y\in \mathbb{Z}_{2^s}^k$ with $f(x)\neq f(y)$, we have \\
\textbf{Case 1}: If $d_h(x,y)\geq 2t+1$, then 
$$d_h(Enc(x),Enc(y))\geq d_h(x,y)+d_h(x_p,y_p)\geq 2t+1.$$
\textbf{Case 2}: If $d_h(x,y)\leq 2t$, then $\tau_f^h(x)\neq\tau_f^h(y)$.  Thus, $x_p\neq y_p$. Consequently,
\begin{align*}
    \begin{split}
d_h(Enc(x),Enc(y))=d_h(x,y)+d_h(x_p,y_p)\geq 2t+1.  
    \end{split}
\end{align*}
By definition of FCCHD, $Enc$ defines an FCCHD with redundancy $2t$.
\end{proof}

\begin{theorem}\label{reN_h}
    Let $f:\mathbb{Z}_{2^s}^k\to \text{Im}(f)$ be a locally $(\lambda,\rho)_h$-function satisfying the contiguous block condition in Lemma \ref{tau}. Then $r^h_f(k,t)\leq N_h(\lambda,2t)$, where $N_h(\lambda
    ,2t)$ is the minimum length of a code with $\lambda$ codewords with minimum homogeneous distance $2t$.
\end{theorem}
\begin{proof}
 By Lemma \ref{tau}, there exists a map $\tau_f^h:\mathbb{Z}_{2^s}^k\to [\lambda]$ such that for all $x,y\in\mathbb{Z}_{2^s}^k$ with $f(x)\neq f(y)$ and $d_ h(x,y)\leq 2t$, we have $\tau_f^h(x)\neq \tau_f^h(y)$. Let $C$ be a code of length $N_h(\lambda,2t)$ with $\lambda$ codewords and minimum homogeneous distance $2t$. We enumerate the codewords of $C$ as $C_1,C_2,\dots,C_{\lambda}$. Let 
 \begin{align*}
     Enc:\mathbb{Z}_{2^s}^k\to \mathbb{Z}_{2^s}^{k+N_h(\lambda,2t)}
 \end{align*}
 be an encoding defined by $Enc(x)=(x,x_p)$, where $x_p=C_{\tau_f^h(x)}$. Now we prove that the set $Enc(\mathbb{Z}_{2^s}^k)$ is an $(f,t)$-FCCHD with redundancy $r=N_h(\lambda,2t)$. Let $x,y\in \mathbb{Z}_{2^s}^k$ with $f(x)\neq f(y)$. If $d_h(x,y) \geq 2t+1$, then 
 \begin{align*}
     d_h(Enc(x),Enc(y))=d_h(x,y)+d_h(x_p,y_p)\geq 2t+1.
 \end{align*}
Else, if $d_h(x,y)\leq 2t$, then  $\tau_f^h(x)\neq \tau_f^h(y)$. That is, $x_p\neq y_p$. Thus,
\begin{align*}
    d_h(Enc(x),Enc(y))=d_h(x,y)+d_h(x_p,y_p)\geq 2t+1.
\end{align*}
By definition of FCCHD, $Enc$ define an $(f,t)$-FCCHD with redundancy $N_h(\lambda,2t)$. Hence, $r_f^h(k,2t)\leq N_h(\lambda,2t)$. 
\end{proof}

From the preceding theorem, the problem of determining the optimal redundancy of an $(f,t)$-FCCHD for locally $(\lambda,2t)_h$-functions reduces to determining $N_h(\lambda,2t)$. In the next theorem, we establish an upper bound on $N_h(\lambda,2t)$ by presenting an explicit construction of codes comprising  $\lambda$ codewords with minimum homogeneous distance at least $2t$. 

\begin{theorem}\label{bound_lambdat}
    Let $N_h(\lambda, 2t)$ be the minimum possible length of a code over $\mathbb{Z}_{2^s}^k$ with $\lambda$ codewords and minimum homogeneous distance $2t$. Then, 
    \begin{align*}
        N_h (\lambda, 2t) \leq \begin{cases}
            \frac{\lambda t}{2} & \text{if } t \text{ is even}\\
            \frac{\lambda (t+1)-2}{2} & \text{if } t \text{ is odd}.
        \end{cases}
    \end{align*}
\end{theorem}
\begin{proof}
Suppose $t$ is even. Let $\bm{0}=(0,0,\dots,0)$ and $\bm{a}=(2^{s-1},2^{s-1},\dots,2^{s-1})$ be vectors of length $\frac{t}{2}$ over $\mathbb{Z}_{2^s}$. Observe that $w_h(\bm{a})=t$.  Let $C=\{C_1,C_2,\dots,C_{\lambda}\}$ be a code over $\mathbb{Z}_{2^s}$ of length $\frac{\lambda t}{2}$, where $$C_i=(\bm{0},\bm{0},\dots, \underbrace{\bm{a}}_{i^{\mathrm{th}}\text{-place}},\bm{0},\dots,\bm{0})\in (\mathbb{Z}_{2^s}^{\frac{t}{2}})^{\lambda},$$ that is, $C_1=(\bm{a},\bm{0},\dots,\bm{0})$, $C_2=(\bm{0},\bm{a},\dots,\bm{0})$, and $C_{\lambda}=(\bm{0},\bm{0},\dots,\bm{a})$. Observe that the homogeneous distance of $C$ is $d_h(C)=2t$. Thus, $N_h(\lambda,2t)\leq \frac{\lambda t}{2}$. \\
Let $t$ be odd. Define a homogeneous code $C=\{C_1,C_2,\dots,C_{\lambda}\}$ of length $\frac{\lambda (t+1)-2}{2}$, where
\begin{align*}
    \begin{split}
        &C_1=(\underbrace{2^{s-1},2^{s-1},\dots,2^{s-1}}_{\frac{t-1}{2}\text-{coordinates}},\underbrace{0,0,\dots,0}_{\frac{(\lambda-1)(t+1)}{2}-\text{coordinates}})\\
        &C_2=(\underbrace{0,0,\dots,0}_{\frac{t-1}{2}\text{-coordinates}},\underbrace{2^{s-1},2^{s-1},\dots,2^{s-1}}_{\frac{t+1}{2}-\text{coordinates}},\underbrace{0,0,\dots,0}_{\frac{(\lambda-2)(t+1)}{2}-\text{coordinates}})\\
        &C_3=(\underbrace{0,0,\dots,0}_{t-\text{coordinates}},\underbrace{2^{s-1},2^{s-1},\dots,2^{s-1}}_{\frac{t+1}{2}-\text{coordinates}},\underbrace{0,0,\dots,0}_{\frac{(\lambda-3)(t+1)}{2}-\text{coordinates}})\\
         &C_4=(\underbrace{0,0,\dots,0}_{\frac{3t+1}{2}-\text{coordinates}},\underbrace{2^{s-1},2^{s-1},\dots,2^{s-1}}_{\frac{t+1}{2}-\text{coordinates}},\underbrace{0,0,\dots,0}_{\frac{(\lambda-4)(t+1)}{2}-\text{coordinates}})\\
         &\vdots\\
     \text{for } i\geq 2, \ \   &C_i=(\underbrace{0,0,\dots,0}_{\frac{(i-2)(t+1)+(t-1)}{2}-\text{coordinates}},\underbrace{2^{s-1},2^{s-1},\dots,2^{s-1}}_{\frac{t+1}{2}-\text{coordinates}},\underbrace{0,0,\dots,0}_{\frac{(\lambda-i)(t+1)}{2}-\text{coordinates}}) \\
         &\vdots\\
         &C_{\lambda}=(\underbrace{0,0,\dots,0}_{\frac{(\lambda -2)(t+1)+(t-1)}{2}-\text{coordinates}},\underbrace{2^{s-1},2^{s-1},\dots,2^{s-1}}_{\frac{t+1}{2}-\text{coordinates}}).
    \end{split}
\end{align*}
Note that $d_h(C_i,C_j)\geq 2t$, $i\neq j$ and $d_h(C_1,C_2)=2t$. Therefore, $d_h(C)=2t$. Hence $N_h(\lambda,2t)\leq \frac{\lambda (t+1)-2}{2} $.
\end{proof}
Using Theorem \ref{reN_h} and Theorem \ref{bound_lambdat}, we have the following bound on the optimal redundancy for locally $(\lambda,2t)_h$-functions.

    \begin{corollary}\label{optimal-upperbound}
    Let $t$ be a positive integer. For a locally $(\lambda, 2t)_h$-function $f$ satisfying the contiguous block condition in Lemma \ref{tau}, the optimal redundancy of $(f, t)$-FCCHDs is bounded above by
    \begin{align*}
        r_f^h (k, t) \leq \begin{cases}
            \frac{\lambda t}{2} & \text{if } t \text{ is even}\\
            \frac{\lambda (t+1)-2}{2} & \text{if } t \text{ is odd}.
        \end{cases}
    \end{align*}
    \end{corollary}
  
    The above bound is tight for $\lambda=2,3$ as shown in the subsequent results.
\begin{corollary}\cite[Theorem 5.1]{liu2025function}
    Let $f:\mathbb{Z}_{2^s}^k\to \text{Im}(f)$ be a locally $(2,2t)_h$-function. Then, the optimal redundancy of an $(f,t)$-FCCHD is $r_f^h(k,t)=t$.
\end{corollary}
% \begin{proof}
%  By putting $\lambda=2$ in Corollary \ref{optimal-upperbound}, we get the desired result.
% \end{proof}

% The bound on optimal redundancy in Corollary \ref{optimal-upperbound} is tight for $(3,2t)_h$-functions over $\mathbb{Z}_{4}^k$.
    
\begin{theorem}\label{optimality_attained3}
    Let  $f:\mathbb{Z}_{4}^k\to \text{Im}(f)$ be a locally $(3,2t)_h$-function satisfying the contiguous block condition in Lemma \ref{tau}. If $|\text{Im}(f)| \geq 3$ and there exist $x_1, x_2, x_3 \in \mathbb{Z}_{2^s}^k$ with distinct function values and 
    $d_h(x_1, x_2) =1, d_h(x_3, x_1)=1 \ \text{and} \ d_h(x_3, x_2)=2$.
     Then the optimal redundancy of an $(f,t)$-FCCHD is \begin{align*}
        r_f^h (k, t)= \begin{cases}
            \frac{3 t}{2} & \text{if } t \text{ is even}\\
            \frac{3 (t+1)-2}{2} & \text{if } t \text{ is odd}.
        \end{cases} 
        \end{align*}
    \end{theorem}
    \begin{proof}
     For given $x_1, x_2, x_3 \in \mathbb{Z}_{4}^k$ with assumptions, the distance requirement matrix in the homogeneous metric is
    $$D_f^h(t, x_1, x_2, x_3) = \left[\begin{matrix} 0 & 2t & 2t \\
    2t & 0& 2t-1 \\
    2t &2t-1 & 0  
    \end{matrix}\right].$$
  By Theorem \ref{plotkinimp4}, we have 
    \begin{align*}
    N_h(D_f^h(t, x_1, x_2, x_3)) &\geq \frac{1}{3^2-1} \sum_{1\leq i,j\leq 3}[D]_{ij}  \\
    & = \frac{1}{8} (12t-2). 
    \end{align*}
\text{Case 1}: Let $t=2\ell$ for some positive integer $\ell$. Then
  \begin{align*}
    N_h(D_f^h(t, x_1, x_2, x_3)) &\geq 
   \frac{1}{8} (24\ell-2)=3\ell-\frac{2}{8}. 
    \end{align*}
\text{Case 2}: Let $t=2\ell+1$ for some positive integer $\ell$. Then
  \begin{align*}
    N_h(D_f^h(t, x_1, x_2, x_3)) &\geq 
   \frac{1}{8} (24\ell+12-2)=3\ell+\frac{10}{8}. 
    \end{align*}   
    
As $N_h(D_f^h(t, x_1, x_2, x_3))$ is an integer, therefore
\begin{align*}
      N_h(D_f^h(t, x_1, x_2, x_3)) \geq \begin{cases}
            3\ell & \text{if } t=2\ell\\
            3\ell+2 & \text{if } t=2\ell+1.
        \end{cases}
    \end{align*}

By Corollary \ref{nbnh}, we have 
    \begin{align*}
    r_f^h(k,t) \geq \begin{cases}
            3\ell & \text{if } t=2\ell\\
            3\ell+2 & \text{if } t=2\ell+1
        \end{cases}=\begin{cases}
            \frac{3 t}{2} & \text{if } t \text{ is even}\\
            \frac{3 (t+1)-2}{2} & \text{if } t \text{ is odd}.
        \end{cases}
         \end{align*} 
    The desired result follows from Corollary \ref{optimal-upperbound}.
    \end{proof}

The homogeneous weight and homogeneous weight distribution function satisfy the contiguous block condition in Lemma \ref{tau}. Thus, we have the following.

\begin{corollary}
Let $f(x)=w_h(x), x\in \mathbb{Z}_{2^s}^k$ be the homogeneous weight function over $\mathbb{Z}_{2^s}^k$. Then there is an $(f,t)$-FCCHD with redundancy 
\begin{align*}
=\begin{cases}
            \frac{\left(4t+1\right) t}{2} & \text{if } t \text{ is even}\\
            \frac{\left(4t+1\right) (t+1)-2}{2} & \text{if } t \text{ is odd}.
        \end{cases}
\end{align*}
\end{corollary}
\begin{proof}
 By Corollary \ref{proplocalhwt}, $f$ is a locally $(4t+1,2t)_h$-function. The construction of $(f,t)$-FCCHD with desired redundancy follows from  Theorem \ref{bound_lambdat} and Theorem \ref{reN_h}.
\end{proof}
 \begin{corollary}
    Let $k, t, T$ be positive integers and $\Delta_T^h$ be the homogeneous weight distribution function over $\mathbb{Z}_{2^s}^k$ with threshold $T$.  Then there exists a $(\Delta_T^h,t)$-FCCHD with redundancy 
    \begin{align*}
        = \begin{cases}
            \frac{\left(\left \lfloor\frac{4t}{T}\right \rfloor+2\right) t}{2} & \text{if } t \text{ is even}\\
            \frac{\left(\left \lfloor\frac{4t}{T}\right \rfloor+2\right) (t+1)-2}{2} & \text{if } t \text{ is odd}.
        \end{cases}
    \end{align*}
\end{corollary}
\begin{proof}
    By Corollary \ref{lamhwd}, $\Delta_T^h$ is a locally $\left(\left \lfloor\frac{4t}{T}\right \rfloor+2,2t\right) $.  The construction of $(f,t)$-FCCHD with desired redundancy follows from  Theorem \ref{bound_lambdat} and Theorem \ref{reN_h}.
\end{proof}

\subsection{Modular sum function}
A modular sum function $ms:\mathbb{Z}_{2^s}^k\to \mathbb{Z}_{2^s}$ is defined as $ms(x)=\sum_{i=1}^kx_i$ for all $x=(x_1,x_2,\dots,x_k)\in \mathbb{Z}_{2^s}^k$. For modular sum function $f$, we have $E=|\text{Im}(f)|=2^s$. 
\begin{proposition}\label{redun-ms}
    Let $ms:\mathbb{Z}_{2^s}^k\to\mathbb{Z}_{2^s}$ be the modular sum function. Then $$r_{ms}^h(k,t)\geq 2t-2^{-s}(2t+1).$$
\end{proposition}
\begin{proof}
Let $M=2^s$. Let $\{x_1,x_2,\dots,x_M\}\subseteq \mathbb{Z}_{2^s}^k$, where $x_i=(0,0,\dots,0,i-1)$ for all $1\leq i\leq M$. Note that $ms(x_i)=i-1$ and  $d_h(x_i,x_j)=d_h(i-1,j-1)$ $\forall i,j$. Thus, the homogeneous distance requirement matrix $D_{ms}^h(t,x_1,x_2,\dots,x_M)$ is given by 
\begin{align*}
    [D_{ms}^h(t,x_1,x_2,\dots,x_M)]_{ij}=\begin{cases}
        2t+1-d_h(i-1,j-1), & \text{ if } i\neq j\\
        0 & \text{ otherwise.}
    \end{cases}
\end{align*}
Consequently,
\begin{align*}
    \begin{split}
       \sum_{1\leq i,j\leq M}[D_{ms}^h(t,u_1,u_2,\dots,u_M)]_{ij}&=(M^2-M)(2t+1)-\sum_{1\leq i,j\leq M}d_h(i-1,j-1)\\
       &=(M^2-M)(2t+1)-\sum_{i=0}^{M-1}\sum_{j=0}^{M-1} w_h(i-j)\\
       &=(M^2-M)(2t+1)-M^2=2tM^2-2tM-M.
    \end{split}    
\end{align*}
By Lemma \ref{plotkin}, we have 
\begin{align*}
    N_h(D_{ms}^h(t,x_1,x_2,\dots,x_M))\geq \frac{1}{M^2}(2tM^2-2tM-M).
\end{align*}
Further, by Corollary \ref{nbnh}, we get the desired result.
\end{proof}
Next, we present a construction for the modular sum function.\\
\textbf{Construction\label{cons1}}: Let $f(x)=ms(x)$ be the modular sum function over $\mathbb{Z}_{2^s}^k$. Let $t$ be a positive integer. Define an encoding $Enc:\mathbb{Z}_{2^s}^k\to \mathbb{Z}_{2^s}^{k+2t}$, $Enc(x)=(x,x_p^{2t})$, where $x_p^{2t}$ is $2t$-fold repetition of $x_p$ and 
\begin{align*}
    x_p=\begin{cases}
        2f(x), & \text{ if } 0\leq f(x)\leq 2^{s-1}-1\\
        2f(x)+1, & \text{ if } 2^{s-1}\leq f(x)\leq 2^{s}-1.
    \end{cases}
\end{align*}
It is easy to see that $x_p\neq y_p$ whenever $f(x)\neq f(y)$ for $x,y\in \mathbb{Z}_{2^s}^k$. Therefore, for $x,y\in \mathbb{Z}_{2^s}^k$ with $f(x)\neq f(y)$, we have 
\begin{align*}
    d_h(Enc(x),Enc(y))=d_h(x,y)+d_h(x_p^{2t},y_p^{2t})\geq 2t+1.
\end{align*}
Hence, encoding map $Enc$ defines an $(f,t)$-FCCHD.\\
By the above construction and Proposition \ref{redun-ms}, we have the following lower and upper bounds on the optimal redundancy for the modular sum function.
\begin{corollary}
    Let $ms(x)$ be the modular sum function over $\mathbb{Z}_{2^s}^k$. Then
    \begin{align*}
        2t-2^{-s}(2t+1)\leq r_{ms}^h(k,t)\leq 2t.
    \end{align*}
 Moreover, if $t<\frac{2^s-1}{2}$, then $r_{ms}^h(k,t)=2t$.   
\end{corollary}

\section{Linear function-correcting codes in the homogeneous metric}\label{FCC-linear}
In \cite{Premlal2024} and \cite{Sampath2025}, the authors derived bounds on the optimal redundancy of function-correcting codes for linear functions in the Hamming and the $b$-symbol metric, respectively. In this section, we investigate function-correcting codes in the homogeneous metric (FCCHD) for linear functions defined over $\mathbb{Z}_{2^s}$. We establish an explicit Plotkin-type upper bound for this class of codes. Additionally, we discuss function-correcting linear codes over  $\mathbb{Z}_{2^s}$.

A  function $f:\mathbb{Z}_{2^s}^k\to \mathbb{Z}_{2^s}^l$ is said to be linear function if the following hold
\begin{align*}
    f(x+y)=f(x)+f(y) \text{ and } f(\alpha\cdot x)=\alpha f(x)\ \ \ \  \forall x,y\in \mathbb{Z}_{2^s}^k, \alpha\in \mathbb{Z}_{2^s}. 
\end{align*}
 The kernel of a linear function $f:\mathbb{Z}_{2^s}^k\to \mathbb{Z}_{2^s}^l$ is defined as
 $$\ker(f)=\{x\in \mathbb{Z}_{2^s}^k| f(x)=0\}.$$

\begin{theorem}\label{plotlinear}
 Let $m=2^s$ and $f:\mathbb{Z}_m^k\to \mathbb{Z}_m^\ell$ be an onto linear function. Then the optimal redundancy of an $(f,t)$-FCLMC is bounded by 
    \begin{align*}
        r_f^h(k,t)\geq (2t+1)(1-2^{-s\ell})-k+\frac{A}{2^{sk}},
    \end{align*}
    where  $A=\sum_{u\in \ker(f)}w_h(u)$.
\end{theorem}
\begin{proof}
  Let $M=m^k$, $\mathbb{Z}_m^k=\{u_1,u_2,\dots,u_M\}$, and $D=D_f^h(t,u_1,u_2,\dots,u_M)$. Let $C=\{C_1,C_2,\dots,C_M\}\subseteq \mathbb{Z}_m^r$ be a $D$-code of length $r$, that is, $d_h(C_i,C_j)\geq [D]_{ij}$. Then, by Lemma \ref{plotkin}, we have 
  \begin{equation}\label{eq7.1}
     M^2r\geq\sum_{i,j}d_L(C_i,C_j)\geq \sum_{i,j}[D]_{ij}.
  \end{equation}
  % where \begin{align*}
  %     S=\begin{cases}
  %         \frac{m^2}{4}, & \text{ if } m \text{ is even}\\
  %           \frac{m^2-1}{4}, & \text{ if } m \text{ is odd.}
  %     \end{cases}
  % \end{align*}
Since $f$ is an onto linear function, therefore cardinality of $\ker(f)$ is $m^{k-\ell}$. Thus, each coset of $\ker(f)$ in $\mathbb{Z}_m^k$ contains $m^{k-\ell}$ elements. Observe that if $u_i,u_j$ belong to the same coset, then $f(u_i)=f(u_j)$. Consequently, $[D]_{ij}=0$, that is, each column of $D$ contains at least $m^{k-\ell}$ zeros. When $u_i,u_j$ are in different cosets, then $[D]_{ij}=[2t+1-d_h(u_i,u_j)]^+$. Let $0\neq v\in \mathbb{Z}_m^k$, $u_{i'}=u_i+v$ and $u_{j'}=u_j+v$. Then $d_h(u_i,u_j)=d_h(u_{i'},u_{j'})$. Therefore, each column (row) in $D$ contains the same set of entries, implying that the columns are permutations of one another. Thus, 
\begin{align*}
  \sum_{i,j} [D]_{ij}=(\text{no. of columns)}\times (\text{sum of one column)}.  
\end{align*}
Let $u_1=\bm{0}$ be the zero vector in $\mathbb{Z}_m^k$. We calculate the sum of the entries of the column corresponding to $u_1$. Let $I=\mathbb{Z}_m^k\setminus \ker(f)$. If $u_i\in \ker(f)$, then $[D]_{i1}=0$. For $i\in I$, $[D]_{i1}\geq 2t+1-d_h(u_i,\bm{0})=2t+1-w_L(u_i)$. Thus,
\begin{align*}
   \sum_{i}[D]_{i1}\geq (2t+1)(m^k-m^{k-\ell})-\sum_{i\in I} w_h(u_i). 
\end{align*}
The sum of homogeneous weights of all the vectors in $\mathbb{Z}_m^k$ is $km^{k}$. Thus, \begin{align*}
    \sum_{i\in I}w_h(u_i)=km^{k}-A, \text{ where } A=\sum_{u\in \ker(f)}w_h(u). 
\end{align*} 
Therefore,
\begin{align}\label{eq7.2}
\begin{split}
    \sum_{i,j} [D]_{ij}&\geq M\cdot  \sum_{i}[D]_{i1}\\
      &\geq m^k((2t+1)(m^k-m^{k-\ell})-km^{k}+A.
    \end{split}
\end{align}
The desired result follows from Eqs. \ref{eq7.1} and \ref{eq7.2}.
\end{proof}

\begin{corollary}
    Let $f:\mathbb{Z}_{2^s}^k\to \mathbb{Z}_{2^s}^\ell $ be a bijective linear function. Then an $(f,t)$-FCCHD reduces to a systematic classical error-correcting code with 
    \begin{align*}
       n= r_f^h(k,t)+k\geq (2t+1)\left(\frac{2^{sk}-1}{2^{sk}}\right).
    \end{align*}
\end{corollary}
\begin{proof}
Since $f$ is bijective, therefore, $\ker(f)=\{\bm{0}\}$. Consequently, $A=0$. The rest follows from Theorem \ref{plotlinear}.    
\end{proof}

\subsection{Function-correcting linear codes in homogeneous metric}
In this subsection, we discuss function-correcting linear codes over $\mathbb{Z}_{2^s}$ for the homogeneous metric.
\begin{definition}
A code $C$ of length $n$ over $\mathbb{Z}_{2^s}$ is said to be a linear code if it is a $\mathbb{Z}_{2^s}$-submodule of $\mathbb{Z}_{2^s}^n$.
\end{definition}

Let $C$ be a linear code over $\mathbb{Z}_{2^s}$ of length $n$. We call a matrix $G_{m\times n}$ a generator matrix for $C$ if the rows of $G$ generate $C$. In \cite{Calderbank1995}, Calderbank and Sloane presented a form of generator matrix for linear codes over $\mathbb{Z}_{2^s}$.
Let $C$ be a linear code of length $n$ over $\mathbb{Z}_{2^s}$. Then  generator matrix for $C$ is given by 
\begin{align*}
    G=\begin{bmatrix}
        I_{k_0} & A_{01} & A_{01} & \dots & A_{0s-1} & A_{0s}\\
        \bm{0} & 2I_{k_1} & 2A_{12} & \dots & 2A_{1s-1} & 2A_{1s}\\
        \bm{0} &\bm{0} & 2^2I_{k_2}& \dots & 2^2A_{2s-1} & 2^2A_{2s}\\
        \vdots &\vdots &\vdots &\vdots &\vdots &\vdots\\
        \bm{0} &\bm{0} &\bm{0} &\dots& 2^{s-1}I_{k_{s-1}} & 2^{s-1}A_{s-1s}
    \end{bmatrix},
\end{align*}
where $A_{ij}$  are matrices over $\mathbb{Z}_{2^s}$ and the columns are grouped into blocks of sizes $k_0,k_1,\dots, k_{s-1}$. Let $k=\sum_{i=0}^{s-1}(s-i)k_i$. Then $2$-dim$(C)=k$, that is, $|C|=2^k$.

\begin{definition}
Let $f:\mathbb{Z}_{2^s}^k\to \text{Im}(f)$ be a function. An $(f,t)$-function-correcting code $C$ with homogeneous distance is called a linear FCCHD with redundancy $r$ if $C$ is a linear code of length $k+r$ over $\mathbb{Z}_{2^s}$. 
\end{definition}

% \begin{lemma}
%     Let $f:\mathbb{Z}_{2^s}^k\to \mathbb{Z}_{2^s}^\ell $ be a linear function and $C$ be a linear $(f,t)$-FCCHD. Consider the subcode 
%     $C_{\ker}=\{c=(u,p(u))\in C\ | \ u\in \ker(f)\}$. Then $C_{\ker}$ is a linear code. Moreover, $w_h(c)\geq 2t+1$ for all $c\in C\setminus C_{\ker}$.
% \end{lemma}
% \begin{proof}
% Let $c_1=(u,p(u))$, $c_2=(v,p(v))\in C_{\ker}$. Then $u,v\in \ker(f)$. Since $\ker(f)$ is submodule of $\mathbb{Z}_{2^s}^k$, therefore $au+bv\in \ker(f)$ for any $a,b\in \mathbb{Z}_{2^s}$. Thus, $ac_1+bc_2=(au+bv,ap(u)+bp(v))\in C_{\ker}$. Hence, $C_{\ker}$ is a linear code.  
% Let $c\in C\setminus C_{\ker}$. Then $c=(u,p(u))$ for some $u\notin \ker(f)$ and $f(u)\neq f(0)$. By definition of $(f,t)$-FCCHD, $d_h(c,\bm{0})=w_h(c)\geq 2t+1$.
% \end{proof}

    Let  $f:\mathbb{Z}_{2^s}^k\to \mathbb{Z}_{2^s}^\ell$ be a linear function and  $C$ be a linear code of length $r$ with a generator matrix $G_{\ell \times r}$. Define 
    $$C_f=\{(u,f(u)\cdot G)\ |\ u\in \mathbb{Z}_{2^s}^k\}.$$
     It is easy to see that if the homogeneous distance of $C$ satisfies $d_h(C)\geq 2t$, then code $C_f$  defines an $(f,t)_h$-function-correcting linear code with redundancy $r$. Note that the construction will not give a function-correcting linear code if the function $f$ is not linear or the code $C$ is not linear.

\begin{example}
Let $s=2$ and $f:\mathbb{Z}_{2^s}^3\to \mathbb{Z}_{2^s}^2$ defined by $f(u_1,u_2,u_3)=(u_1+u_2,u_2+u_3)$ for all $(u_1,u_2,u_3)\in \mathbb{Z}_{2^s}^3$. Let $C$ be a linear code over $\mathbb{Z}_{2^s}$ generated by matrix 
\begin{align*}
    G=\begin{bmatrix}
        2 & 2 & 0\\
        0 & 2 &2
    \end{bmatrix}.
\end{align*}
 Note that $d_h(C)=4$. Then 
$C_f=\{(u,f(u)\cdot G)\ |\ u\in \mathbb{Z}_{2^s}^k\}$
is a linear code of length $6$ and
defines an $(f,2)_h$-function correcting linear code with redundancy $3$. 
\end{example}

\section{Conclusion}\label{sectionconclusion}
In \cite{liu2025function}, a general framework of function-correcting codes in the homogeneous metric over $\mathbb{Z}_{2^s}$ is introduced. In this paper, we studied function-correcting codes for weight functions, modular sum function, and families of functions, namely, locally bounded and linear functions in the homogeneous metric over $\mathbb{Z}_{2^s}$. First, we analyzed the locality behavior of the homogeneous weight function and its associated weight-distribution functions. We then established lower and upper bounds on the optimal redundancy for locally bounded functions and provided explicit constructions.  For linear functions, we derived an explicit Plotkin-type bound for function-correcting codes in the homogeneous metric, showing that it reduces to the classical Plotkin bound when the linear function is bijective, and we further discussed constructions of function-correcting linear codes over $\mathbb{Z}_{2^s}$. In \cite{Rajput2026}, function-correcting codes with data protection were introduced in the Hamming-metric setting. A natural and promising direction is to extend this notion to the homogeneous metric, where the underlying ring structure of $\mathbb{Z}_{2^s}$ may yield new constructions, sharper bounds, and potentially more efficient constructions for simultaneous function recovery and data protection.

\bibliographystyle{abbrv}
	\bibliography{ref}

\end{document}